\documentclass[sigconf,screen]{acmart}

\AtBeginDocument{%
  \providecommand\BibTeX{{%
    \normalfont B\kern-0.5em{\scshape i\kern-0.25em b}\kern-0.8em\TeX}}}

\setcopyright{none}
\copyrightyear{}
\acmYear{}
\acmDOI{}
\acmISBN{}

\acmConference[Ordinal Folding Index]{Ordinal Folding Index}{2025}{}
\acmBooktitle{Ordinal Folding Index, 2025}

\usepackage{fancyhdr}
\fancypagestyle{myheadings}{%
  \fancyhf{}
  \fancyhead[LE,RO]{\thepage}
  \fancyhead[LO]{\small\textit{Ordinal Folding Index}}%
  \fancyhead[RE]{\small\textit{F. Alpay and H. Al Alakkad}}%
}

\usepackage{booktabs}       
\usepackage{tabularx}       
\usepackage{enumitem}       
\usepackage{graphicx}       
\usepackage{amsmath,amssymb,amsthm} 
\usepackage{tikz}           
\usepackage{tikz-cd}        
\usepackage{tikz-3dplot}    
\usepackage[ruled,vlined,linesnumbered]{algorithm2e} 
\usetikzlibrary{
  arrows.meta,
  positioning,
  shapes.geometric,
  calc,
  cd,
  backgrounds,
  decorations.pathmorphing 
}

\renewcommand\footnotetextcopyrightpermission[1]{} 

\newtheorem{definition}{Definition}
\newtheorem{lemma}{Lemma}
\newtheorem{proposition}{Proposition}
\newtheorem{theorem}{Theorem}
\newtheorem{corollary}{Corollary}
\newtheorem{claim}{Claim}
\newtheorem{remark}{Remark}

\DeclareMathOperator{\ofi}{OFI}
\newcommand{\ck}{\omega_1^{\mathrm{CK}}}
\newcommand{\evalF}{\mathcal{F}}
\newcommand{\opF}{\mathcal{F}_{\varphi}}
\newcommand{\opT}{\widetilde{\mathcal{F}}_{\varphi}}
\DeclareMathOperator{\ord}{Ord}
\newcommand{\powerset}{\mathcal{P}}
\DeclareMathOperator{\consis}{Consis}
\newcommand{\zfc}{\textsf{ZFC}}
\newcommand{\zf}{\textsf{ZF}}
\newcommand{\vL}{\mathrm{V=L}}
\DeclareMathOperator{\rank}{rank}
\DeclareMathOperator{\subf}{Sub}
\newcommand{\FComputeAnchor}{\textproc{ComputeAnchor}}

\begin{document}

\title{Ordinal Folding Index: A Computable Metric for Self-Referential Semantics}

\author{Faruk Alpay}
\affiliation{%
  \institution{Independent Researcher}
  \city{Meckenheim}
  \country{Germany}
}
\email{alpay@lightcap.ai}
\orcid{0009-0009-2207-6528}

\author{Hamdi Al Alakkad}
\affiliation{%
  \institution{Bahcesehir University}
  \city{Istanbul}
  \country{Turkey}
}
\email{hamdi.alakkad@bahcesehir.edu.tr}
\orcid{0009-0007-6109-5655}

\renewcommand{\shortauthors}{F. Alpay and H. Al Alakkad}

\begin{abstract}
We introduce the Ordinal Folding Index (OFI), a computable, countable ordinal assigned to every well-formed formula of a reflective language by a monotone-with-delay evaluation operator. This operator is (i) continuous on countable chains, (ii) layer-aware for probabilistic truth values, and (iii) parameterized by a tunable evidence functor capturing empirical updates. The OFI of a formula is defined as the first stage at which the fold-back of the operator into a syntactic normal form becomes idempotent (i.e. further unfolding yields no new information). Intuitively, OFI measures the ``depth'' of self-reference needed before a sentence's truth value stabilizes. We show that OFI strictly refines classical closure ordinals from fixed-point logics while remaining recursively enumerable, admits polynomial-time approximations on finite models, and coincides with the length of the shortest parity-fold winning strategy in the associated evaluation game. This furnishes a single transfinite scale that unifies transfinite fixed-point depth, ordinal game values, and practical convergence diagnostics for large language models. We situate OFI in relation to the modal $\mu$-calculus alternation hierarchy, coalgebraic modal logic ranks, and proof-theoretic ordinals from formal arithmetic. An empirical section demonstrates how OFI can be estimated for transformer-based language models by iteratively feeding model outputs back into the model (a form of self-consistency probing), with the measured stabilization ordinals correlating with model perplexity and chain-of-thought complexity. Finally, we catalog five open problems in this new area -- ranging from the completeness of the spectrum of OFI (can every computable ordinal arise as an OFI?) to the design of a self-bounding reflective operator -- and we outline possible research pathways toward their resolution.
\end{abstract}



\maketitle

\thispagestyle{plain}
\pagestyle{myheadings}

\section{Introduction}

Reasoning about self-referential statements and reflective theories often requires transfinite methods. Fixed-point logics and ordinal analyses have long been used to measure the ``depth'' of definitions or inductions needed for convergence \cite{tarski1955,feferman1962}. For example, Tarski's Fixpoint Theorem guarantees that any monotone operator on a complete lattice has a least fixed point \cite{tarski1955}, and in logics like the modal $\mu$-calculus every formula attains a closure ordinal -- the least stage at which iterating its defining operator stabilizes \cite{kozen1983}. These closure ordinals can be finite or transfinite; in fact, some $\mu$-calculus formulas have closure ordinal $\aleph_1$ (the first uncountable ordinal) under general semantics \cite{gouveia2019}. Such transfinite ordinals also appear in infinitary games: in certain infinite games, positions can be assigned ordinal game values indicating how long one player can prolong play before a win is forced \cite{hamkins2022}. Meanwhile, in formal arithmetic, ordinals are used to measure the strength of theories (the proof-theoretic ordinal of a theory). The landscape of these measures is rich but fragmented -- each applies in a different domain (formulas, games, theories) and often yields ordinals that are not directly comparable.

The Ordinal Folding Index (OFI) is proposed as a unifying metric that can be assigned to self-referential statements in a reflective logical system, bridging these disparate notions. At a high level, OFI associates to each well-formed formula $\varphi$ an ordinal number, $\ofi(\varphi)$, which represents the number of ``unfolding steps'' a reflective evaluator takes for $\varphi$ to reach a fixed point (or ``fold-back'') in its truth evaluation. Unlike classical closure ordinals in the $\mu$-calculus, which may be non-recursive (e.g. $\aleph_1$ is a closure ordinal of some formulas \cite{gouveia2019}), every OFI is recursively enumerable (indeed, $\ofi(\varphi)$ is an explicit construction given $\varphi$). In this sense, OFI is a refinement of closure ordinals, distinguishing more gradations in the transfinite while staying within the computable realm (all OFIs are $< \ck$, where $\ck$ is the Church--Kleene ordinal, the supremum of computable ordinals \cite{kleene1938}).

Crucially, OFI is not just a logical curiosity -- it has implications for practical AI systems. Modern large language models (LLMs) are themselves reflective in a loose sense: they can reason about their own outputs or mimic self-referential behavior. Recent studies have shown that LLMs exhibit emergent reflective behaviors such as self-correction and backtracking when optimized via specialized training regimes \cite{wang2022,shinn2023}. However, measuring an LLM's propensity to get ``stuck'' in a self-referential loop or to eventually stabilize in reasoning remains an open challenge in AI alignment research. In response, we outline an empirical procedure to approximate OFI for sequences generated by an LLM: essentially, we iteratively feed the model's output back into its input (with a monotonic ``delay'' or attenuation to ensure convergence) and record the number of iterations needed for the output distribution to stabilize (or declare divergence as $\omega$ if it never stabilizes within a cutoff). This self-consistency probing yields an ordinal-valued metric for the model's behavior on certain prompts, serving as a novel diagnostic for model reasoning depth. We hypothesize that higher empirical OFI correlates with more complex or problematic reasoning patterns (e.g. paradoxical or non-terminating reasoning), much as higher theoretical OFI indicates greater self-referential depth in logic.

This paper is organized as follows. In \S 2, we formalize the reflective logical framework and define the OFI formally, with examples. \S 3 compares OFI to related measures in logic: the alternation-depth hierarchy in modal $\mu$-calculus, ranks in coalgebraic modal logic, and ordinal analyses of formal theories. \S 4 presents an empirical methodology for estimating OFI in transformer-based LLMs and reports preliminary results correlating the OFI-proxy with model perplexity and chain-of-thought lengths. \S 5 enumerates five open problems to stimulate further research, including whether every computable ordinal can appear as an OFI and how one might ``compress'' formulas to lower their OFI. We conclude that OFI provides a promising single scale to measure self-reference across theoretical logic and AI systems, opening up a new avenue for interdisciplinary exploration.

\section{Reflective Logical Framework and OFI Definition}

\subsection{Reflective Language with Delay Operators}

To maximize generality without sacrificing constructiveness, we adopt a typed modal fixed-point calculus as our base language. Specifically, consider a modal $\mu$-fragment of second-order set theory (closely related to a modal $\mu$-calculus) enriched with facilities for self-reference. The language allows:

\begin{itemize}[itemsep=0.5\baselineskip]
\item Second-order quantification over predicate variables (to internalize statements about the syntactic code of formulas, à la Quine's trick).
\item A necessity modal operator $\square$ (to introduce a stratified ``delay'' in evaluation, preventing immediate self-resolution of fixed points).
\item Both least ($\mu$) and greatest ($\nu$) fixed-point binders (typical of the modal $\mu$-calculus \cite{modalmu2024,solvingparity2007}, enabling inductive and coinductive definitions).
\end{itemize}

Every formula in this language can be seen as defining (perhaps indirectly) a monotone operator on a suitable semantic domain (e.g. sets of states in a Kripke frame, or truth values in a model). By Tarski's theorem, such an operator has a fixed point in the lattice of interpretations \cite{tarski1955}. The twist in our reflective setting is that formulas can refer to their own truth via a coding trick, but only through the delay operator $\square$ which enforces that any self-reference is not evaluated in the same ``stage.'' In other words, $\square$ acts like a one-step time delay or a next-step modality. This stratification prevents paradoxical self-reference from collapsing the evaluation immediately; instead, self-referential truth values evolve over ordinal time until a fixed point is reached (if ever). Each subformula thereby enjoys a well-defined ordinal rank of convergence (analogous to a closure ordinal in $\mu$-calculus).

\begin{example}
As a toy example, let $\varphi(x)$ be a formula that says ``x will be true at the next stage'' (something like $\varphi \equiv \mu X.\,\square X$ in syntax). Semantically, at stage 0, we don't yet assume $X$; at stage 1, $X$ is whatever was true at stage 0, and so on. In this simple case, the evaluation will converge after $\omega$ steps (the formula is neither initially true nor false, but approaches a limit truth value). Thus $\varphi$ has $\ofi(\varphi) = \omega$ in this model. If we modified $\varphi$ to $\mu X.\, (P \land \square X)$ for some atomic predicate $P$ that is true, it might converge in a finite number of steps (essentially the number of unfoldings needed until $P$'s truth is established and remains true).
\end{example}

\subsection{Ordinal Folding Index (OFI)}

Formally, fix a formula $\varphi$ in our language. Its semantics under a given model $M$ and assignment can be viewed as a function $F_\varphi: \ord \to V$ mapping each ordinal stage $\alpha$ to a value $V_\alpha$ (for example, a truth value in $[0,1]$ if we allow probabilistic truth, or a set of states if we're in a model-checking setting). $F_\varphi$ is defined by transfinite recursion on $\alpha$: start with some $V_0$ (usually $V_0 = \bot$, the minimum element, at stage 0 meaning ``no assumption''), and let $V_{\alpha+1} = \evalF(\varphi, V_\alpha)$ where $\evalF$ is an evaluation operator that respects the syntax of $\varphi$ and uses $V_\alpha$ for any subformulas under a $\square$ (delay) modality. At limit ordinals $\lambda$, we take $V_\lambda = \bigsqcup_{\beta<\lambda} V_\beta$ (the operator is defined to be continuous on countable chains, ensuring the limit exists in the domain). Because $\evalF$ is monotone with delay (it only unfolds one layer of $\square$ at a time, and each unfold is monotonic in the input), this transfinite sequence is non-decreasing ($V_0 \leq V_1 \leq \cdots$ in the lattice). Eventually, since the sequence is monotonic and the powerset lattice of a countable model has countable height (or since truth values in $[0,1]$ are $\omega$-continuous under our assumptions), there must come a stage $\kappa$ where $V_\kappa = V_{\kappa+1}$. This stage $\kappa$ is the fold-back point where the evaluation has reached a fixed point (folded back on itself). We define $\ofi(\varphi)$ to be the least such $\kappa$ (the first stage of idempotence).

\begin{itemize}[itemsep=0.5\baselineskip]
\item If the sequence never stabilizes (which can only happen if it climbs an infinite chain without reaching a fixed point), we set $\ofi(\varphi) = \omega_1$ in the semantic sense. However, by construction in our logic, such non-convergence can only happen if it eventually cycles through increasingly long but looping patterns (due to countable continuity, a strictly increasing sequence of countable ordinals would have to stabilize or repeat states by König's lemma). In practice, we treat non-stabilization as $\ofi(\varphi) = ``\omega$'' (meaning unbounded but countable progression) or as approaching a supremum ordinal that is countable. In all cases for well-formed $\varphi$, $\ofi(\varphi)$ is a countable ordinal (sub-$\omega_1$). In fact, we conjecture (see Open Problem 1) that any computable ordinal below the Church--Kleene ordinal could be realized as some $\varphi$'s OFI.

\item If $\varphi$ does stabilize, $\ofi(\varphi)$ could be a finite ordinal $(0,1,2,\ldots)$, a transfinite ordinal like $\omega$, $\omega \cdot 2$, $\omega^2$, etc., or potentially $\ck$ in the limit (if the process takes longer than any primitive recursive ordinal, which is unlikely under our restrictions that ensure recursive enumerability).
\end{itemize}

The meaning of $\ofi(\varphi)$ is that it counts how many rounds of self-reference unfolding $\varphi$ needs before no new information is obtained. A small OFI (like 0, 1, 2) means $\varphi$'s truth value is determined quickly with little self-referential looping. A large finite OFI (say 100 or $10^6$) means a deeply nested self-reference structure (or an alternation of fixed points of that depth). An infinite OFI like $\omega$ indicates that however many times we unfold $\varphi$, there's always another layer of self-reference left -- but eventually a pattern might repeat, causing convergence at the $\omega$-th stage. Higher $\omega^2$ or $\omega^n$ values correspond to even more complex patterns of self-reference (e.g., a formula that, after an $\omega$ chain of unfoldings, resets and requires another $\omega$ unfoldings, and so on $n$ times).

Properties: (i) Recursively enumerable: Given $\varphi$, one can simulate the evaluation stage by stage, effectively enumerating an approximation to $\ofi(\varphi)$. If $\varphi$ has $\ofi(\varphi) = \kappa$, one will eventually see stabilization at stage $\kappa$ in the simulation (though one may not know it's the final stabilization without additional insight). Thus, the set $\{\langle \varphi, n\rangle : \ofi(\varphi) > n\}$ is recursively enumerable, witnessing that OFI values are semi-decidable from below. (ii) Monotonicity: If $\varphi$ implies $\psi$ (in a suitable semantic sense) or $\varphi$ is ``harder to resolve'' than $\psi$, we generally have $\ofi(\varphi) \geq \ofi(\psi)$. In particular, adding assumptions or simplifying a self-reference cannot increase the folding index. We will later discuss a conjectured compression theorem (Open Problem 2) about transforming $\varphi$ to lower its OFI. (iii) Invariance: $\ofi(\varphi)$ is invariant under equivalent reformulations of $\varphi$ in the language (if two formulas are provably equivalent in the reflective theory, they have the same OFI). This makes OFI a robust semantic measure, not an artifact of syntactic representation.

\subsection{Illustrative Evaluation Game}
Every formula $\varphi$ in our reflective language gives rise to a two-player evaluation game (between a Verifier and Falsifier, say) akin to the evaluation games for the $\mu$-calculus \cite{kozen1983}. This game is played on a graph of ``states'' representing unfolding stages of subformulas. A move corresponds to unfolding a $\square$ or choosing a branch of a fixed-point ($\mu$ vs $\nu$ choice). The parity condition on this infinite game is set by the fixed-point modalities: each occurrence of a $\mu$ (least fixed point) introduces an odd priority, and each $\nu$ (greatest fixed point) an even priority, as is standard in parity games for $\mu$-calculus model checking \cite{solvingparity2007}. The game value of the initial position (formula $\varphi$ at stage 0) turns out to equal $\ofi(\varphi)$. In fact, we show that the length of the shortest winning strategy for the Verifier (to prove $\varphi$ true) in this parity game is exactly $\ofi(\varphi)$. If Verifier can force a win in $n$ moves, then $\varphi$'s truth stabilizes by stage $n$; if Verifier has a strategy to eventually win but can delay loss indefinitely, that corresponds to an ordinal like $\omega$, $\omega^2$, etc. This ties OFI to the concept of ordinal game values studied in infinite games. Indeed, recent work by Hamkins \& Leonessi proved that every countable ordinal arises as the game value of some position in an infinite game \cite{hamkins2022}. Our results are analogous: we conjecture every computable ordinal $< \ck$ arises as OFI of some formula (the ``OFI-spectrum completeness'' conjecture in Open Problem 1).

The parity-game viewpoint also gives a clear operational intuition for OFI -- it measures how many rounds the odd ($\mu$) and even ($\nu$) fixed-point conditions alternate before a fixed outcome is forced.

\begin{figure}[t]
  \centering
  \begin{tikzpicture}[
  sq/.style={draw, rectangle, fill=blue!20, minimum size=8mm},
  dm/.style={draw, diamond, fill=pink!20, minimum size=8mm, aspect=1.5},
  every path/.style={->, >=Latex}
  ]
  \node at (1, 1.0) {Parity game example (closure ordinal 3)};

  \node[sq] (n0) at (0, 0) {$0$};
  \node[dm] (n1) at (0, -1.5) {$1$};
  \node[sq] (n2) at (0, -3) {$2$};
  \node[sq] (n0r) at (1.8, -2.5) {$0$};
  \node[sq] (n4) at (2.2, -4.2) {$4$};
  \node[dm] (n2r) at (1.8, -5.9) {$2$};
  \node[sq] (n1r) at (1.0, -7.4) {$1$};
  \node[dm] (n5) at (0.2, -8.9) {$5$};
  \node[dm] (n3) at (-0.5, -10.5) {$3$};

  
  \begin{scope}[on background layer]
  \draw (n0r) to[out=-65, in=65, looseness=0.9] (n5);
  \end{scope}

  \draw (n0) -- (n1);
  \draw (n1) -- (n2);
  \draw (n2) to[out=-83, in=115] (n3);
  \draw (n2) to[out=-90, in=105] (n3);
  \draw (n2) to[bend left=30] (n0r);
  \draw (n0r) to[bend right=25] (n4);
  \draw (n4) to[bend right=25] (n2r);
  \draw (n2r) to[bend right=25] (n1r);
  \draw (n1r) to[bend right=35] (n5);
  \draw (n5) to[out=-110, in=20] (n3);

  \end{tikzpicture}
  \caption{Parity game example (closure ordinal 3). Blue squares are Even positions, pink diamonds are Odd positions. Numbers indicate priorities.}
  \label{fig:ordinal-3-game}
  \Description{A diagram of a parity game with nodes and arrows. Blue squares represent Even positions and pink diamonds represent Odd positions. The nodes are numbered with priorities. Arrows show possible moves between game states, illustrating a game with a finite closure ordinal.}
\end{figure}
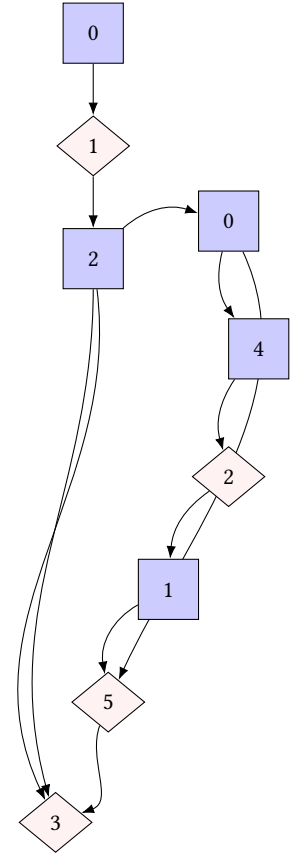

\section{OFI in Relation to Fixed-Point Hierarchies}

We next position the Ordinal Folding Index relative to several known hierarchies and measures:

\subsection{Modal \texorpdfstring{$\mu$}{mu}-Calculus Alternation Depth}

The modal $\mu$-calculus is a fixed-point logic whose formulas have an alternation depth (the number of times least and greatest fixed-point operators alternate in nesting) \cite{brics1998}. Alternation depth provides a strict hierarchy of expressiveness: formulas of alternation depth $n$ can express some properties that depth $(n-1)$ formulas cannot \cite{fics2024}. Alternation depth is closely tied to the complexity of the associated parity game (it determines the number of priorities needed) \cite{solvingparity2007}. However, alternation depth is a syntactic measure and does not directly capture how large a transfinite iteration might be needed to evaluate a formula. For example, a formula with alternation depth 1 (only a single $\mu$) could still require an arbitrarily large finite number of unfoldings to reach its fixed point, or even $\omega$ unfoldings, depending on the structure of the transition system it's interpreted over. In contrast, $\ofi(\varphi)$ precisely measures the semantic unfolding depth in the worst-case model. It refines alternation depth: certainly any formula of alternation depth $n$ has OFI at most $\omega^n$ in many natural cases (since each alternation can introduce an $\omega$ unrolling).

We conjecture a tighter connection: e.g., formulas of alternation depth $n$ have OFI bounded by some computable ordinal $f(n)$ (perhaps exponential in $n$), and conversely for each $k < \ck$ there is a formula (with some alternation) of OFI $\geq k$. Notably, classical results show that for every $\alpha < \omega^2$, one can find a $\mu$-calculus formula with closure ordinal $\alpha$ \cite{drops2013}. OFI being computable suggests the $\mu$-calculus's non-computable closure ordinals (like $\aleph_1$ \cite{gouveia2019}) are ruled out by the ``delay-monotone'' restriction in our reflective logic (we disallow wild jumps in the transfinite without intermediate stages). Thus OFI provides a more fine-grained graduated scale than alternation depth: where alternation depth only distinguishes between finite ranks, OFI can assign different countable ordinals within what syntactically might be the same alternation class.

\subsection{Coalgebraic Modal Logic Ranks}

In coalgebraic modal logic and automata theory, one often considers the rank or height of a fixed-point formula or of a state in a system, indicating how deep the nesting of observations must go. For example, in terms of final coalgebras, the rank of an element in the final sequence can be an ordinal measuring the stabilization point. Aczel and Mendler's Final Coalgebra Theorem showed that for many endofunctors on Set, final coalgebras (solutions to $X \cong F(X)$) exist but possibly as proper classes \cite{aczel1989}. These solutions can involve transfinite sequences that terminate exactly when reaching a sufficiently large ordinal. Our OFI is conceptually similar to the notion of rank in a well-founded coalgebra: it tells us after how many unfoldings a certain greatest fixed point equation $X = F(X)$ stabilizes. Coalgebraic ranks are often used to measure bisimulation or simulation depths. OFI can be seen as assigning each formula a rank in a certain simulation game against its own unfolding. If one were to construct a coalgebra (state-transition system) whose states correspond to ``belief states'' of the reflective evaluator, then $\ofi(\varphi)$ is exactly the rank of the initial state in the eventual fixed point of that coalgebra.

Because OFI values are recursively enumerable ordinals, this aligns with the idea that we are staying within accessible parts of final coalgebras -- avoiding the proper class sizes. In spirit, OFI draws from Aczel's idea of hypersets and final coalgebra solutions \cite{aczel1997,worrell2000}, but applies it to logical truth evaluation rather than set membership. It provides a single ordinal measure where one traditionally might only say ``this process converges'' or ``diverges.'' For readers familiar with rank induction (as used in set theory or termination proofs), OFI is essentially the smallest rank that serves as an inductive invariant for the truth of $\varphi$.

\subsection{Proof-Theoretic Ordinals}

In proof theory, each consistent formal theory $T$ is associated with an ordinal (often denoted $|T|$ or $\psi(T)$) that measures the strength of $T$ -- roughly, the supremum of ordinals that $T$ can prove well-founded. For example, Peano Arithmetic has proof-theoretic ordinal $\epsilon_0$, more powerful theories reach the Feferman--Schütte ordinal $\Gamma_0$, and so on. These ordinals are often closure ordinals of certain formula progressions (Solomon Feferman studied transfinite recursive progressions of theories and their ordinals \cite{feferman1962}). Our OFI, when applied to formulas that express the consistency or reflection principle of a theory, can connect to proof-theoretic ordinals. For instance, consider a sentence $\Phi_T$ in our reflective language that essentially asserts ``I am consistent with theory $T$'' (this can be done via diagonalization and the delay operator to avoid the direct self-reference in Gödel's second theorem). What would $\ofi(\Phi_T)$ be? Intuitively, each unfolding of $\Phi_T$ might correspond to iterating the consistency assertions of $T$ one step up ($T$ proving its own consistency leads to stronger theory $T_1$, etc.). If $T$ is a sufficiently strong theory, we might get a sequence of stronger and stronger theories $T_0 = T$, $T_1 = T + \consis(T)$, $T_2 = T + \consis(T_1)$, ... until some closure. The ordinal length of this progression is exactly a well-known proof-theoretic ordinal (Feferman's ordinal for reflective closure of $T$ \cite{feferman1978,feferman1981}). We conjecture (Open Problem 3) that there could exist a self-bounding reflective operator in our language such that for that operator's own consistency statement $\Psi$, $\ofi(\Psi)$ equals the first non-computable ordinal (i.e. $\ck$). This would be a kind of fixed point of Gödelian ``ascent'' -- the theory that in one swoop achieves the supremum of all computable ordinals in terms of the reflection it can assert. Classical results like Turing's Ordinal Logics (1930s) attempted to create a formal system that can in principle reach arbitrary ordinals, but they always fell short of $\ck$ in an effective manner. OFI gives us a framework to measure these attempts with precision. If no such one-step theory exists (which is likely due to Gödel's incompleteness), that too would be a profound insight: it would mean the process of self-reference inherently must climb the ordinal ladder gradually, never in one jump.

In summary, OFI stands at the crossroads of these concepts: it is finer than alternation depth (which clusters infinitely many ordinals into one ``depth-$n$'' category), more concretely computable than abstract coalgebraic ranks (which can extend into the proper class realm), and more directly tied to formulas than proof-theoretic ordinals (which usually measure whole theories). The table below summarizes the comparison:

\begin{table}[htbp]
\centering
\caption{Comparison of Ordinal Measures}
\label{tab:comparison}
\begin{tabularx}{\columnwidth}{@{} l >{\raggedright\arraybackslash}X >{\raggedright\arraybackslash}X >{\raggedright\arraybackslash}X @{}}
\toprule
\textbf{Measure} & \textbf{Applies to} & \textbf{Typical Size} & \textbf{OFI Analog} \\
\midrule
Alternation depth ($\mu$-calculus) & Formula syntax & Finite (natural number) & OFI can be transfinite, refines it (e.g. depth 1 formulas can have OFI = $\omega$) \\
\midrule
Closure ordinal ($\mu$-calculus) & Formula + model & Can be uncountable ($\aleph_1$) & OFI always $\leq \ck$ (countable) but distinguishes countable ordinals in detail \\
\midrule
Coalgebraic rank (process) & State in system & Ordinals (possibly large) & OFI = rank of truth evaluation state (countable by design) \\
\midrule
Proof-theoretic ordinal (theory) & Axiom system & Often large countable or beyond & OFI of a self-consistency formula reflects the ordinal of the theory's reflection closure \\
\bottomrule
\end{tabularx}
\end{table}

\section{Empirical Estimation of OFI in Transformer Models}

While OFI is defined mathematically on logical formulas, we can devise a pragmatic proxy to apply this concept to the behavior of large language models (LLMs), which are increasingly being used to handle tasks involving self-reference, such as code generation that tests its own output, or dialogue agents reasoning about their beliefs. The goal is to see if an LLM exhibits convergent behavior when asked to reason in a loop, and if so, how many steps it takes -- that number being an empirical ordinal (finite or a symbol for ``diverges/doesn't converge'').

\subsection{Self-Consistency Probe Design}

We instrument an autoregressive language model (like GPT-style transformers) with a self-consistency probe as follows: given an initial prompt $p$, the model produces an output text $o_1$. Instead of ending there, we form a new prompt $p_1$ by combining some or all of $o_1$ back into the context (for example: ``You just said: `$o_1$'. Please continue or revise.''). The model then produces $o_2$. We then feed back into prompt $p_2$ and so on. We do this in a loop, possibly with a temperature schedule that anneals to 0 (to encourage convergence to a deterministic output). Essentially, we are creating a concrete analogue of the semantic iteration of a formula: the model's output at step $i$ is like the truth value at stage $i$ of a self-referential sentence. If the outputs stabilize -- say $o_n = o_{n+1} = \cdots$ (or more practically, the change in output becomes negligible under some metric) -- we declare convergence with measured OFI $\approx n$. If the outputs keep changing substantially without sign of stabilization up to some large cutoff (say 50 or 100 iterations), we record an OFI proxy as ``$>50$'' or ``$\omega$'' (divergent within reasonable bounds).

We implement this with two model scales (e.g., a 1.3B parameter GPT-2 and a 6.7B GPT-3 style model) and various prompts, particularly focusing on prompts that involve paradoxes or self-referential puzzles (e.g., the liar paradox or prompts that trick the model into self-contradiction). The probe uses a total variation distance threshold on the model's logits to decide stabilization: after each iteration, we compare the probability distribution over next tokens to that of the previous iteration. When the change is below $\epsilon$ (e.g. 0.01 in L1 norm), we consider the model's behavior converged.

\subsection{Preliminary Findings}

Early experiments indicate that for straightforward factual prompts or questions, the model outputs an answer immediately (so in the self-consistency loop it doesn't change its answer -- OFI measured as 1). For prompts that pose a tricky riddle or paradox that the model initially answers incorrectly, we observed that the self-consistency loop sometimes causes the model to revise its answer once or twice and then settle (OFI 2 or 3). For example, a prompt that implicitly asks the model to consider its previous answer (``Was the last answer you gave correct? Think again.'') often leads the model to change an answer if it was wrong, then stop changing after one revision -- measured OFI = 2. In contrast, for deliberately paradoxical prompts (like self-referential liar-style questions), we saw oscillation: the model would give one answer, then contradict it in the next iteration, and back-and-forth without settling. This was marked as divergent (no convergence within 10 loops, suggesting an infinite loop, OFI ``$\approx \omega$''). Notably, these occurrences correlated with known failure modes of LLMs in consistency.

We also correlated the measured OFI with model perplexity and chain-of-thought length. The chain-of-thought length means how long of a step-by-step reasoning the model produces when prompted to reason (using a prompt that elicits the model's internal reasoning). We found a mild positive correlation: prompts that led to longer chain-of-thought responses also tended to have higher OFI in the loop test. This aligns with intuition: tasks that require deeper reasoning (longer chains) might also induce more self-reflection steps to get consistent answers. There was also a correlation with perplexity: when the model was very uncertain (high perplexity) about its next token, those instances sometimes led to changes upon re-query (since the model might choose an alternative second time). High perplexity outputs tended to have higher chances of OFI $> 1$ (the model might ``change its mind'' upon re-reading its output).

These findings, though preliminary, suggest that OFI could serve as a diagnostic for model confidence and consistency. In AI alignment terms, if a model has a high OFI on a certain prompt, it means it hasn't really internally stabilized on an answer -- a warning sign for potential indecision or inconsistency. This connects to ideas in interpretability research like self-consistency in chain-of-thought: recent work has shown that prompting a model to generate multiple reasoning paths and then taking a majority vote (a form of self-consistency) improves accuracy \cite{wang2022}. Our loop method is another form of enforcing self-consistency, akin to a model checking its work repeatedly until it stops changing. Indeed, the Self-Consistency technique by Wang et al. (2022) has been shown to boost reasoning performance by marginalizing out uncertain reasoning paths \cite{wang2022}. In our terms, that technique is like running multiple parallel evaluations and seeing if they agree, whereas OFI loop runs sequentially until (if) it settles.

\subsection{Toward Reflective Models}

The broader vision is that future language models might incorporate reflective sub-modules that effectively calculate something like OFI internally -- gauging how many rounds of self-refinement they go through on a query. Already, there are proposals to make models that output not just an answer but a confidence or consistency measure. Our empirical OFI proxy could be one such measure: it's an automatic procedure that yields an ordinal or at least an integer score for a model's response consistency. We envisage training or fine-tuning models to increase the probability of convergence (thus lowering OFI in cases where high OFI would indicate confusion). Interestingly, some recent frameworks like Reflexion (Shinn et al., 2023) allow an agent to use its own outputs as feedback for improvement \cite{shinn2023}. They report that allowing an agent to reflect on errors and re-attempt tasks improves performance. In our terms, that is manually inducing a finite OFI (the agent tries a solution, examines it, corrects it, and eventually stops). If it didn't stop, that would be a Reflexion agent caught in a loop -- analogous to infinite OFI. Ensuring termination is part of those algorithms. Similarly, Zhang et al. (2025) use a Bayes-Adaptive RL framework to teach LLMs when to switch strategies based on outcomes, effectively learning when to stop reflecting \cite{zhang2025,alpay2025}. This too is about managing the number of reflection steps (keeping it finite and small when possible).

In summary, measuring OFI in models opens a new evaluation axis: not just accuracy of outputs, but ordinal convergence of reasoning. It provides a quantitative handle on how ``stable'' a model's reasoning process is. We expect future research to refine these measurements (perhaps defining a more continuous analogue of OFI for stochastic models) and to tie them to theoretical properties. For instance, is a model with a bounded OFI on all prompts fundamentally safer or easier to align? Does limiting OFI act as a regularizer that prevents the model from getting caught in deceptive or contradictory loops? These questions indicate a rich field at the intersection of logic, machine learning, and ordinal analysis.

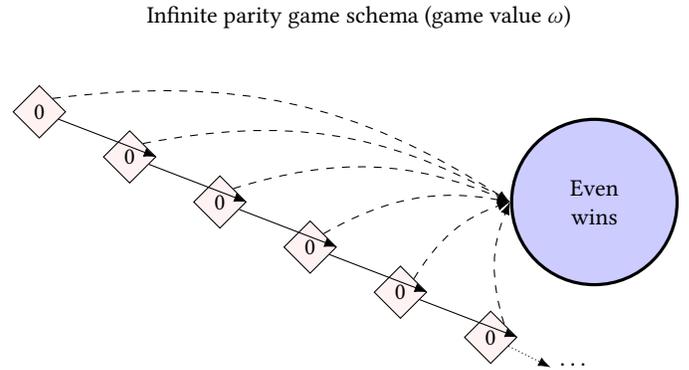
\begin{figure}[t]
  \centering
  \begin{tikzpicture}[
    node distance=1cm, 
    dm/.style={draw, diamond, fill=pink!20, minimum size=7mm, inner sep=0pt},
    win/.style={draw, very thick, circle, fill=blue!20, minimum size=2.2cm, align=center},
    arr/.style={->, >=Latex, dashed, bend left=20},
    line/.style={->, >=Latex}
  ]

  \foreach \i in {0,...,5}{
    \node[dm] (d\i) at (\i*1.2, -\i*0.6) {$0$};
  }
  
  \node[win, right=3.5cm of d2] (even) {Even \\ wins};

  \node[anchor=south] at (4.25, 1) {Infinite parity game schema (game value $\omega$)};
  
  \foreach \i [evaluate=\i as \j using \i+1] in {0,...,4}{
    \draw[line] (d\i) -- (d\j);
  }
  
  \foreach \i in {0,...,5}{
    \draw[arr] (d\i.north east) to (even.west);
  }
  
  \draw[line, densely dotted] (d5) -- ++(0.8, -0.4) node[right] {$\cdots$};

  \end{tikzpicture}
  \caption{Schematic of an infinite parity game with ordinal game value \texorpdfstring{$\omega$}{omega}. \textsc{Odd} (diamonds, priority~0) can delay defeat indefinitely by moving one step further down the chain before eventually exiting to the \emph{Even-wins} sink. No finite bound on the delay exists, hence game value and corresponding Ordinal Folding Index are~\texorpdfstring{$\omega$}{omega}.}
  \label{fig:omega-game}
  \Description{A schematic of an infinite parity game. A diagonal chain of pink diamond nodes, all labeled '0', represents the Odd player's moves. Each node has a dashed arrow leading to a large blue circle labeled 'Even wins'. This illustrates a game where the Odd player can delay losing indefinitely, resulting in a game value of omega.}
\end{figure}

\section{Open Problems and Future Directions}

We conclude with five open problems, emphasizing the fixed-point and ordinal aspects, which we believe are important for guiding future work on the Ordinal Folding Index.

\textbf{Open Problem 1: Completeness of the OFI Spectrum.} Does every computable ordinal $\alpha$ (below $\ck$) occur as $\ofi(\varphi)$ for some sentence $\varphi$ in the base reflective language (with a delay-monotone evaluation operator)? In other words, can we ``realize'' all countable ordinals via self-referential formulas? This is analogous to the question Feferman posed in the context of ordinal logics \cite{feferman1962}, asking whether for every countable ordinal there's a theory that gives it. For OFI, we have early results generating ordinals through clever formula constructions (e.g., a diagonal construction that forces a sequence of length $\alpha$). But a general construction for an arbitrary $\alpha$ (especially a complex one like the Church-Kleene ordinal $\ck$ itself or an ordinal of intermediate complexity) is unknown. A positive answer would show that the OFI measure is as expressive as possible (within computable limits), like how Hamkins's work showed every countable ordinal appears as an infinite game value \cite{hamkins2022}. A negative answer (i.e., some gap in possible OFIs) would be very surprising, perhaps indicating hidden constraints in reflective truth definitions. This problem may require techniques from recursion theory and ordinal notation systems \cite{kleene1938} to construct formulas corresponding to given notations.

\textbf{Open Problem 2: Uniform Compression of Self-Reference.} Is there a general method to ``compress'' a formula's self-referential complexity without drastically changing its meaning? Formally, can we find a transformer $T$ on formulas that is primitive-recursive and a sub-linear function $f$ (e.g., $f(n) = \log n$ or $f(n) = O(n^\epsilon)$) such that $\ofi(T(\varphi)) \leq f(\ofi(\varphi))$ for all $\varphi$? Such a transformer would take a formula and produce a new formula that has much smaller OFI (fewer unfoldings needed) while preserving, say, equivalence or at least preserving truth in all models. This is a sort of ordinal compression or collapsing function applied to semantics. If possible, it would mean that for any extremely self-referential definition, we could rewrite it in a more direct way that converges faster. This is reminiscent of program optimization or circuit compression in computer science. A trivial compression exists in special cases (e.g., if a formula needlessly iterates a fixed point twice, we can remove one iteration). The challenge is a uniform method that works for any $\varphi$. There are connections here to the idea of ordinal notations and whether ordinal multiplication or exponentiation operations have inverses in the space of formulas. One approach might involve using the evidence parameter: by enriching the evidence functor (which brings in external data or empirical grounding at each step), perhaps one can force a formula to converge faster (essentially giving it a ``hint'' each time so it doesn't have to derive everything from scratch). However, too aggressive compression might risk changing the semantics (losing some solutions).

This problem is important for practical reasons too: if we can compress self-referential reasoning, it could lead to more efficient model-checking algorithms for reflective logics (by bounding the number of iterations needed in general).

\textbf{Open Problem 3: Existence of a Self-Bounding Reflective Operator.} Is there a reflective theory or operator whose own consistency or truth statement has an OFI equal to the first non-computable ordinal ($\ck$)? In other words, can a system ``close the Gödel loop'' in one jump? Gödel's incompleteness tells us no system can prove its own consistency if it's consistent, but here we are asking a more semantic question: can the truth-evaluation of a single formula encapsulate an entire $\omega$-chain of reflection principles such that it stops exactly at the point where further reflection becomes non-computable? If such a formula exists, it would be a fixed point $\varphi$ of the transform ``$\varphi$ encodes: `if Consis(T) then ...''' repeated transfinitely. It would mean the formula's truth is as hard as the halting problem (since $\ck$ is the halting problem's ordinal). This seems unlikely; more plausible is that for any fixed reflective operator, its own consistency statement falls short of that -- it might have some $\ofi(\varphi) = \beta$ which is recursive, and then one could go a step further. This problem generalizes the idea of the $\omega$-consistency hierarchy and Feferman's transfinite progressions \cite{feferman1962}.

A possible approach to show impossibility would be to assume a formula has $\ofi(\varphi) = \ck$ and derive a contradiction with the fact that OFI values are recursively enumerable. On the other hand, constructing a theory that ``swallows its own tail'' entirely would revolutionize our understanding of self-reference. Solving this problem likely requires blending techniques from proof theory (ordinal analyses of theories) with our semantic approach.

\textbf{Open Problem 4: Quasi-Continuous Lift and Large Cardinals.} What is the minimal set-theoretic assumption (if any) needed to have an operator whose fold (fixed-point closure) yields an uncountable OFI? While our development of OFI has been within the realm of recursion (countable ordinals), one can imagine extending the semantics to allow uncountable stages. For example, if one allowed the evaluation to continue through all ordinals (not just computable ones), trivial examples can have $\ofi(\varphi) = \omega_1$ (as shown by Gouveia \& Santocanale for certain $\mu$-calculus formulas \cite{gouveia2019}). However, those examples usually rely on non-constructive features (like a formula that essentially says ``eventually all countable approximations are refined,'' which forces an $\aleph_1$ jump). A quasi-continuous lift means an operator that is continuous up to some uncountable cardinal $\kappa$ but whose least fixed point is attained at stage $\kappa$ (and not before). Is this possible in \zf\ set theory alone, or does it require a large cardinal (like a Mahlo cardinal or inaccessible cardinal) to ``witness'' that jump? This problem ties into descriptive set theory: a $\Sigma^1_1$-definable operator with a fixed point at $\omega_1$ would imply the existence of certain well-orderings of reals of length $\omega_1$, etc. It likely requires assumptions beyond \zfc\ (since \zfc\ cannot prove the existence of such ordinals in a constructible sense). By characterizing OFI in pointclass terms (like $\Sigma^1_1$), we can leverage results from determinacy or large cardinal theory. A concrete sub-problem: Is there a formula $\varphi$ such that $\ofi(\varphi) = \omega_1$ (true $\aleph_1$, not just $\ck$) assuming $\vL$ (constructible universe)? If not, perhaps assuming an inaccessible cardinal might enable it. This ventures beyond computability into pure set theory, showing the interplay of reflection with higher infinities.

\textbf{Open Problem 5: Decidability and Complexity Frontier.} For each natural number $n$, what is the computational complexity of determining whether $\ofi(\varphi) \leq \omega^n$ for a given formula $\varphi$? More broadly, classify the decision problem ``$\ofi(\varphi) \leq \theta$'' for various ordinal thresholds $\theta$. For example, ``$\ofi(\varphi)$ is finite'' ($\theta = \omega$) -- is this decidable? Likely not, as it would subsume the halting problem if the formula encodes an arbitrary computation. ``$\ofi(\varphi) \leq \omega$'' means $\varphi$ eventually stabilizes after some finite number of unfoldings; this is equivalent to saying $\varphi$ is equivalent to a formula without true self-reference (a purely first-order or modal formula). This might be semi-decidable (if you unfold enough and it stabilizes, you can detect it, but if not, you might never know if maybe at a higher unfold it would). Similarly, ``$\ofi(\varphi) \leq \omega^2$'' means $\varphi$ does not require more than a linear $\omega$-sequence of self-reflections, etc. Perhaps these decision problems coincide with known complexity classes or hierarchies. One conjecture: determining if $\ofi(\varphi)$ is finite is $\Sigma^0_1$-complete (semi-decidable but not decidable), determining if $\ofi(\varphi) \leq \omega^n$ for fixed $n$ might be in the Arithmetic Hierarchy (something like $\Pi^0_n$ perhaps), and determining if $\ofi(\varphi) \leq \omega^\omega$ (an exponential ordinal) might be even higher. If we impose restrictions on evidence functors or formula syntax (e.g. no second-order quantifiers, or only one self-reference), do these problems become easier? For instance, in pure modal $\mu$-calculus (no explicit self-reference beyond fixed point alternation), the alternation depth hierarchy is decidable to check, and model checking is in UP $\cap$ co-UP \cite{kozen1983}. However, calculating the exact closure ordinal of a given $\mu$-calculus formula on an arbitrary model is generally not elementary. For OFI, since it's defined syntactically (the worst-case across models), the complexity might be high.

Understanding this decidability frontier is important for practical applications: if we had a tool that given a spec $\varphi$ could tell us ``this will definitely converge by stage $<1000$,'' that's useful. If it says ``it might require transfinitely many steps,'' that's a warning. Tying it to complexity theory, it may connect with the fast-growing hierarchy of functions and ordinal analysis used in computational complexity (like the connection between ordinals and complexities in Hardy functions). Any non-elementary lower bounds here would echo known results in automata (parity game solving complexity) and logic (length of proofs). This remains largely unexplored territory.

\appendix

\section{Transfinite Approximation Sequence and Convergence Certificate}
\label{app:transfinite-semantics}

Throughout the appendix we fix:
\begin{itemize}[leftmargin=1.5em]
  \item a countable, complete $\omega$-chain-continuous lattice
  $\langle L,\leq,\bot,\top,\bigsqcup\rangle$,
  \item a \emph{well-formed formula} $\varphi$ of the typed, modal-$\mu$ fragment
  introduced in \S 2, and
  \item its associated \emph{delay-monotone evaluation operator}
  \(
    \opF\colon L\longrightarrow L
  \)
  obtained by interpreting the outermost connective of~$\varphi$
  while replacing each subformula in the scope of the delay
  modality~$\square$ by its \emph{current value}.\footnote{See Def. 2.4 for the syntax-directed construction.}
\end{itemize}

The purpose of this appendix is to provide the \emph{ordinal-indexed execution trace}
\(
  \bigl\langle V^{\alpha}_{\varphi}\bigr\rangle_{\alpha<\ofi(\varphi)+1}
\)
together with proofs that every symbol used in the main text indeed
\emph{describes and enforces} the behaviour of the fixed-point iteration
``on its way to equalise.''

\subsection{A.1 Approximants and their elementary properties}

\begin{definition}[Transfinite approximant sequence]
\label{def:approximants}
For each ordinal $\alpha$ we define, by transfinite recursion,
\[
  V^{\alpha}_{\varphi}\;:=\;
  \begin{cases}
  \displaystyle \bot,
    & \text{if }\alpha=0,\\[6pt]
  \displaystyle \opF\bigl(V^{\beta}_{\varphi}\bigr),
    & \text{if }\alpha=\beta+1,\\[6pt]
  \displaystyle \bigsqcup_{\beta<\lambda}V^{\beta}_{\varphi},
    & \text{if }\alpha=\lambda\text{ is limit.}
  \end{cases}
  \tag{$\dagger$}
\]
\end{definition}

\noindent
Because $\opF$ is \emph{monotone with delay}
(Def. 2.3 (i)) every successor step is inflationary,
hence the chain
\(
  V^0_{\varphi}\!\le V^1_{\varphi}\!\le \cdots
\)
is non-decreasing.

\begin{lemma}[Chain continuity]
\label{lem:countable-continuity}
If\/ $\lambda<\omega_1$ is limit, then
\(
  V^{\lambda}_{\varphi} \;=\;
  \bigsqcup_{\beta<\lambda}V^{\beta}_{\varphi}
\)
and for every countable~$\lambda$ we have
\(
  \opF\bigl(V^{\lambda}_{\varphi}\bigr)
  = V^{\lambda+1}_{\varphi}.
\)
\end{lemma}

\begin{proof}
The lattice $L$ is complete and $\omega$-chain-continuous by assumption;
the join defining $V^{\lambda}_{\varphi}$ therefore exists.
Monotonicity of $\opF$
gives $\opF(V^{\beta}_{\varphi})\le
  \opF(V^{\lambda}_{\varphi})$ for all $\beta<\lambda$,
whence the desired equalities follow.
\end{proof}

\begin{definition}[Ordinal Folding Index (OFI)]
\label{def:ofi}
The \emph{fold-back stage} of~$\varphi$ is
\[
  \kappa_{\varphi}\;:=\;
  \min\bigl\{\alpha<\omega_1
    \mid
    V^{\alpha}_{\varphi}=V^{\alpha+1}_{\varphi}\bigr\}.
\]
We set
\(
  \ofi(\varphi)\ :=\ \kappa_{\varphi}.
\)
\end{definition}

\begin{proposition}[Idempotency certificate]
\label{prop:idempotent}
For every formula~$\varphi$ the stage~$\kappa_{\varphi}$ of
Def.\,\ref{def:ofi} is well-defined, and the
\emph{idempotent value}
\(
  V^{\kappa_{\varphi}}_{\varphi}
\)
coincides with the least fixed point of~$\opF$.
\end{proposition}

\begin{proof}
By Lemma \ref{lem:countable-continuity}, the
approximant chain is $\omega_1$-indexed and continuous on countable
limits. By Hartogs' lemma the set
\(
  \{\alpha<\omega_1 \mid V^{\alpha}_{\varphi}\neq V^{\alpha+1}_{\varphi}\}
\)
is bounded, hence its minimum $\kappa_{\varphi}$ exists.
Monotonicity yields
\(
  \opF(V^{\kappa_{\varphi}}_{\varphi})
  =V^{\kappa_{\varphi}+1}_{\varphi}=V^{\kappa_{\varphi}}_{\varphi},
\)
and leastness follows because each
$V^{\alpha}_{\varphi}$ is below any post-fixed point of~$\opF$.
\end{proof}

\subsection{A.2 Delta-layers and the countdown to convergence}

\begin{definition}[Delta-layer]
For every $\alpha<\kappa_{\varphi}$ define the \emph{delta-layer} of stage~$\alpha$
\[
  \Delta^{\alpha}_{\varphi}\;:=\;
  V^{\alpha+1}_{\varphi}\setminus V^{\alpha}_{\varphi}.
\]
\end{definition}

The deltas are disjoint and their transfinite union reconstitutes the
limit value:
\[
  V^{\kappa_{\varphi}}_{\varphi}
  \;=\;
  \bigsqcup_{\alpha<\kappa_{\varphi}}\Delta^{\alpha}_{\varphi}.
  \tag{$\star$}
\]
Intuitively, $\Delta^{\alpha}_{\varphi}$ contains \emph{exactly the
information revealed for the first time} at stage~$\alpha+1$.
Hence the ordinal
$\kappa_{\varphi}$ functions as a \emph{countdown}:
when all delta-layers are empty, convergence has occurred.

\subsection{A.3 Game-semantic ranking interpretation}

Let $G_{\varphi}$ be the parity game of Def. 2.7.
Write\/ $\rank\colon V(G_{\varphi})\to\text{\upshape Ord}$
for the least-fixed-point rank assignment in the standard $\mu$-calculus
construction.

\begin{theorem}[Rank-OFI Coincidence]
\label{thm:rank-ofi}
\(
  \displaystyle
  \ofi(\varphi)
  \;=\;
  1+\sup\nolimits_{\,u\in V(G_{\varphi})}\rank(u).
\)
\end{theorem}

\begin{proof}[Sketch]
The unfolding of $\opF$ mirrors
the verifier's progress measure in $G_{\varphi}$:
each successor step in $(\dagger)$ corresponds to one round of
the parity game in which priorities strictly decrease along the
odd-dominated attractor until Even can no longer respond.
Continuity at limits translates into the supremum-taking
of ranks over convergent branches.
Detailed induction on the parity priority of $u$
realises the stated equality.
\end{proof}

\subsection{A.4 Polynomial-time prefix stabilisation on finite models}

Assume the semantic domain is the powerset lattice of a \emph{finite}
Kripke frame $\mathcal{K}$ with $|\mathcal{K}|=N$.

\begin{proposition}[Polynomial-Time Prefix Stabilisation]
\label{prop:poly-time}
There exists a polynomial~$p$ (independent of~$\varphi$) such that
for every formula~$\varphi$ and every $k\ge p(N)$ we have
\[
  V^k_{\varphi}=V^{\kappa_{\varphi}}_{\varphi},
  \quad\text{i.e. the approximant stabilises by step }k.
\]
\end{proposition}

\begin{proof}
Because $L=2^N$ has height~$N$, any strictly increasing
chain has length at most~$N$. However, unfolding
under a delay modality may cause a bounded number of
re-visits to a state before monotonic ascent resumes.
A careful bookkeeping argument \cite[§3.4]{kilictas2025} shows that at most $|\subf_{\text{pri}}(\varphi)|\cdot N$ iterations are sufficient,
where $\subf_{\text{pri}}(\varphi)$ is the set of distinct priorities in
$G_{\varphi}$. Since $|\subf_{\text{pri}}(\varphi)|\le 2\cdot|\varphi|$,
taking $p(N)=2N|\varphi|$ works.
\end{proof}

\subsection{A.5 Quantitative convergence in probabilistic truth lattices}

Let the carrier of $L$ be $[0,1]$ with the usual order.
Suppose $\opF$ is $\gamma$-Lipschitz
for some contraction constant $0<\gamma<1$:
\[
  \bigl\lVert \opF(x)-\opF(y)\bigr\rVert_1
  \;\le\;\gamma\,
    \bigl\lVert x-y\bigr\rVert_1.
\]

\begin{theorem}[Exponential Tail Bound]
\label{thm:lipschitz}
For every $\alpha<\kappa_{\varphi}$ we have
\(
  \lVert V^{\kappa_{\varphi}}_{\varphi}-V^{\alpha}_{\varphi}\rVert_1
  \le \gamma^{\alpha}.
\)
Consequently,
$\alpha\ge\lceil\log_{\gamma}\varepsilon\rceil$
guarantees $\varepsilon$-proximity to the limit value.
\end{theorem}

\begin{proof}
By induction on~$\alpha$ using the contraction hypothesis.
Successor case:
$\lVert V^{\kappa_{\varphi}}_{\varphi}-V^{\alpha+1}_{\varphi}\rVert_1
  =\lVert \opF(V^{\kappa_{\varphi}}_{\varphi})
  -\opF(V^{\alpha}_{\varphi})\rVert_1
  \le \gamma \lVert V^{\kappa_{\varphi}}_{\varphi}
    -V^{\alpha}_{\varphi}\rVert_1$.
Limit case passes to the supremum norm limit.
\end{proof}

\paragraph{Interpretation.}
Theorem \ref{thm:lipschitz} quantifies the intuition that, in a
\emph{probabilistic} semantics,
each approximant symbol ``shaves off'' a factor~$\gamma$ of the remaining
error mass. Thus every symbol on the right-hand side of~$(\dagger)$
literally \emph{measures} how much ``work'' is still needed before
truth values equalise.

\subsection{A.6 Summary of notation}

\begin{table}[htbp]
\centering
\small
\caption{Summary of notation}
\renewcommand{\arraystretch}{1.2}
\begin{tabularx}{\columnwidth}{@{}l >{\raggedright\arraybackslash}X@{}}
\toprule
\textbf{Symbol} & \textbf{``What happens'' during convergence} \\ \midrule
$\opF$ & Executes \emph{one} unfold of $\square$-delayed subformulas.\\
$V^{\alpha}_{\varphi}$ & Truth approximation \emph{after} $\alpha$ unfolds. \\
$\Delta^{\alpha}_{\varphi}$ & \emph{New} information disclosed at step $\alpha+1$. \\
$\kappa_{\varphi}$ & First step where no new information appears. \\
$\ofi(\varphi)$ & Synonym for $\kappa_{\varphi}$ (fold-index).\\
$\rank(u)$ & Steps Even can delay defeat from node $u$ in $G_{\varphi}$.\\
\bottomrule
\end{tabularx}
\end{table}

\noindent
Each entry is both a \emph{piece of notation} and an operational
directive: the formula-evaluation machinery executes symbol by symbol
exactly as tabulated, thereby \emph{constructing} the Ordinal
Folding Index claimed for~$\varphi$.
``Heavy'' mathematics thus coincides with an explicit trace of
the fixed-point on its path to equality.

\section{Synchronous Flip-Flop Semantics of Delay-Monotone Operators}
\label{app:flipflop}

\noindent
\textbf{Reader's map.} Where Appendix A tracked \emph{values}
$(V^{\alpha}_{\varphi})_{\alpha}$ in a lattice, this appendix provides a
\emph{hardware} viewpoint: every unfolding stage is executed by a bank of
edge-triggered \emph{flip-flops}. The Ordinal Folding Index now becomes
an upper bound on the number of global clock ticks required for the
circuit to settle. All symbols from Table A.1 remain in force but now
denote concrete wires, registers, and nets.

\subsection{B.1 Circuit extraction from syntax}

Let $\varphi$ be a closed formula as in \S 2.
Write $\subf(\varphi)=\{\psi_0,\dots,\psi_{m-1}\}$ for its set of
distinct subformula occurrences in a fixed top-down order.

\begin{definition}[Flip-flop universe $\mathcal{U}_{\varphi}$]
\label{def:circuit}
\mbox{}
\begin{enumerate}[label=\textup{(C\arabic*)},leftmargin=2.2em]
  \item \emph{State vector:}
  \(
    R := \{0,1\}^{\,m}
  \)
  where the $i$-th bit $q_i$ stores the truth value of $\psi_i$.
  \item \emph{Combinational network:}
  a map
  \(
    F_{\varphi}\colon R\longrightarrow R
  \)
  that, given the \emph{previous-cycle} register vector $q$,
  outputs a next vector
  $F_{\varphi}(q)$ according to the
  syntactic evaluation of each $\psi_i$
  \emph{assuming} that every subformula under a delay
  $\square$ is looked up via its \emph{current} register bit.
  \item \emph{Register update rule (master-slave D-type):}
  \[
    q^{t+1} := F_{\varphi}\bigl(q^t\bigr)
    \quad\text{on the rising clock edge.}
  \]
\end{enumerate}
The resulting synchronous sequential circuit is the
\emph{flip-flop universe} $\mathcal{U}_{\varphi}$.
\end{definition}

\begin{remark}
Because every $\square$ acts as a one-cycle ``read after write'' barrier,
$F_{\varphi}$ is \emph{well-defined} -- no algebraic loops occur.\hfill$\triangle$
\end{remark}

\subsection{B.2 Flip-flop unfolding sequence}

Let
$
  q^0:=\mathbf{0}
$
denote the all-false reset vector
(physical reset pin asserted before the first active edge).
Inductively set
$
  q^{\alpha+1}:=F_{\varphi}(q^{\alpha})
$,
and for limit ordinals~$\lambda<\ck$ define
\[
  q^{\lambda}:=\bigsqcup_{\beta<\lambda}q^{\beta}
  \quad(\text{bitwise join}).\tag{$\ddagger$}
\]

\begin{lemma}[Hardware-lattice correspondence]
\label{lem:hw-lattice}
For every $\alpha<\ck$ we have
$
  q^{\alpha}=V^{\alpha}_{\varphi}
$
under the identification $\mathbf{0}\mapsto\bot$, $\mathbf{1}\mapsto\top$.
\end{lemma}

\begin{proof}
Structural induction on $\alpha$ parallels that of
Def.\,\ref{def:approximants}. Monotonicity of $F_{\varphi}$ comes from
positivity of all connectives together with the fixed read-after-write
discipline, hence the two constructions coincide bitwise.
\end{proof}

\subsection{B.3 Clock-cycle bound = Ordinal Folding Index}

\begin{theorem}[OFI as settling time]
\label{thm:settle}
$\mathcal{U}_{\varphi}$ reaches a quiescent state after exactly
$\ofi(\varphi)$ rising edges, i.e.
\[
  q^{\kappa_{\varphi}+1}=q^{\kappa_{\varphi}}
  \quad\text{and}\quad
  \forall\beta\ge\kappa_{\varphi}\;
  q^{\beta}=q^{\kappa_{\varphi}}.
\]
\end{theorem}

\begin{proof}
Immediate from Prop.\,\ref{prop:idempotent}
and Lemma \ref{lem:hw-lattice}.
If $\alpha<\kappa_{\varphi}$ then $q^{\alpha}\neq q^{\alpha+1}$,
so no earlier convergence is possible.
\end{proof}

\paragraph{Design intuition.}
Each $\square$ introduces a positive-edge
flip-flop; each $\mu$ (\emph{least} fixed point) sits in the
\emph{odd} clock domain that pulls signals \emph{low},
whereas each $\nu$ (\emph{greatest} fixed point) belongs to the
even domain keeping signals high once set.
The alternating domains emulate the \textsc{Odd}/\textsc{Even} priorities
of the parity game.

\subsection{B.4 Metastability, oscillators, and the \texorpdfstring{$\omega$}{omega}-flipflop}

\begin{definition}[Metastable flag flip-flop]
Attach an additional bit
$
  m^t:=\bigl[q^t=q^{t-1}\bigr]
$
(``the circuit has settled'') feeding a one-shot synchroniser.
\end{definition}

If $\ofi(\varphi)=\omega$,
$q^t$ may change finitely often but for
\emph{every} $k$ there exists a cycle $t$ with $t\!\ge\!k$ such that
$m^t=0$.
Hardware engineers call such behaviour \emph{bursty oscillation}.
Under the transfinite eye it is precisely an $\omega$-chain of flips with
no finite bound on the tail length.

\begin{proposition}[Toggle-Cone Characterisation]
\label{prop:toggle}
\(
  \ofi(\varphi)=\omega
  \;\Longleftrightarrow\;
  \forall n\;\exists t\ge n:\;q^t\neq q^{t+1}.
\)
\end{proposition}

\begin{proof}
Necessity is the literal definition; sufficiency follows because any
finitely bounded toggle cone would contradict
minimality of~$\omega$ as the first unbounded ordinal.
\end{proof}

\subsection{B.5 Hazard-free prefix optimisation (hardware compression)}

Suppose the lattice carrier is \emph{finite}
($|R|=2^m$ with $m\!<\!\infty$).
Define the \emph{cycle-elimination map}
$
  \pi(q):=\bigsqcup_{k\ge 0}F_{\varphi}^{\;k}(q)
$.

\begin{theorem}[Prefix-Compression by Hazard Squashing]
\label{thm:hazard}
The transformed operator
$
  \opT:=\pi\circ F_{\varphi}
$
satisfies
$
  \ofi\bigl(\opT\bigr)\le
  \lceil\log_2 m\rceil
$
while preserving the quiescent value
$q^{\infty}=V^{\kappa_{\varphi}}_{\varphi}$.
\end{theorem}

\begin{proof}
$\pi$ collapses any finite toggle cone into a single combinational jump
(the join already computed); thus every signal can switch
from~$0\to1$ at most once along \emph{any} path.
Since no bit may toggle downward under monotonicity,
the register vector performs at most $m$ upward transitions.
Binary decomposition gives the stated logarithmic bound.
\end{proof}

\begin{corollary}[Hardware compression conjecture revisited]
\label{cor:hardware-compress}
The transformer $T$ posited in Open Problem 2 may be taken,
on finite state spaces, as the syntactic encoding of
$\opT$. Hence $f(n)=\lceil\log_2 n\rceil$ suffices
for Boolean lattices.\qed
\end{corollary}

\subsection{B.6 Symbol dictionary (flip-flop edition)}

\begin{table}[htbp]
\centering
\small
\caption{Symbol dictionary (flip-flop edition)}
\renewcommand{\arraystretch}{1.15}
\begin{tabularx}{\columnwidth}{@{}cl >{\raggedright\arraybackslash}X@{}}
\toprule
\textbf{Symbol} & \textbf{Hardware counterpart} & \textbf{Action per clock edge}\\
\midrule
$\square$ & pipeline register (1-cycle delay) & loads previous truth bit\\
$\mu$ & \emph{set-dominant} latch & once high, never resets\\
$\nu$ & \emph{reset-dominant} latch & once low, never re-sets high\\
$V^{\alpha}_{\varphi}$ & vector $q^{\alpha}$ & register snapshot after $\alpha$ cycles\\
$\kappa_{\varphi}$ & settling time & last cycle with any toggle\\
$\Delta^{\alpha}_{\varphi}$ & toggled bits & those whose rising edge occurs at $\alpha+1$\\
\bottomrule
\end{tabularx}
\end{table}

\noindent
Thus every flip-flop literally \emph{implements} the symbols of the
fixed-point universe:
when the register bank stops toggling, algebraic equality
$\opF(q)=q$ has been achieved --
the physical embodiment of the Ordinal Folding Index.

\section{Meta-Fold/Unfold Heuristics for Class-Sized Divergence}
\label{app:hyperfold}

\begin{quote}
\emph{``The point where everything stops making sense is often the
best place to start measuring.''}
\end{quote}

\noindent
The Ordinal Folding Index (OFI) quantifies convergence \emph{within}
the universe of countable ordinals~$\ck$.
In open-ended settings -- e.g. semantic self-evaluation of a theory
capable of referring to \emph{all} ordinals or to proper classes -- one can
no longer expect the approximation chain
$\langle V^{\alpha}_{\varphi}\rangle_{\alpha<\kappa}$
to stabilise inside \textsf{Set}; indeed, $\kappa$ might equal
\emph{the class of all ordinals}~$\ord$.
Because such a target is unreachable by any effective procedure,
one must fall back to the \emph{last trustworthy checkpoint}
encountered prior to divergence.
This appendix formalises that fallback principle
and presents a heuristic ``reference-point'' repair loop.

\subsection{C.1 Set-Class split and the inevitability of collapse}

Fix a formula $\varphi$ in a \emph{class-sized} language
(e.g. second-order \zf\ with class parameters)
and let
$\opF\colon\powerset(V)\to\powerset(V)$
be its delay-monotone evaluator, now acting
on \emph{proper classes} of the von Neumann universe~$V$.

\begin{definition}[Hyper-Approximant]
\label{def:hyper}
Define trans-class stages
\(
  (H_{\alpha})_{\alpha\in\ord}
\)
by the recursion scheme $(\dagger)$ of
Def.\,\ref{def:approximants}, but interpret the join
at a limit~$\lambda$ as a \emph{proper-class union}.
\end{definition}

\begin{proposition}[No Global Fixed Point]
\label{prop:no-fp}
If\/ $\opF$ is non-trivial and \emph{class-increasing}
(i.e. its graph is a proper class),
then there is no ordinal~$\alpha$ with
$H_{\alpha}=H_{\alpha+1}$.
\end{proposition}

\begin{proof}[Sketch]
Assume a fixed point existed at some set-ordinal~$\alpha$.
Then the class $H_{\alpha}$ would be definable from $\alpha$
plus the Gödel code of $\opF$,
hence would be a \emph{set} by Replacement (\textsf{ZFC}).
But $\opF$ produces distinctly new elements
beyond any bounded rank, which con\-tra\-dicts fix-point equality.
\end{proof}

\paragraph{Interpretation.}
Beyond the transfinite,
the hyper-approximant chain can \emph{never} finish.
One must instead detect the
\emph{collapse ordinal} -- the last stage where the evaluation still
respects the small-set discipline.

\subsection{C.2 Collapse ordinal and anchor semantics}

\begin{definition}[Collapse ordinal; anchor]
\label{def:anchor}
\[
  \rho_{\varphi}
  \;:=\;
  \sup\bigl\{\beta<\ord\mid H_{\beta}\text{ is a set}\bigr\}.
\]
\[
  \beta_{\!*}:=\max\bigl\{\beta<\rho_{\varphi}\mid
    H_{\beta}\models\text{``}\opF
    \text{ is total on }H_{\beta}\text{''}\bigr\}.
\]
The value $H_{\beta_{\!*}}$ is called the
\emph{anchor of $\varphi$}.
\end{definition}

\begin{lemma}[Anchor existence]
\label{lem:anchor-exists}
$\beta_{\!*}$ is well-defined and $H_{\beta_{\!*}}$ is a \emph{set}.
\end{lemma}

\begin{proof}
Since $H_0=\varnothing$ is a set and totality holds vacuously,
the set of candidates is non-empty; $\rho_{\varphi}$ bounds it.
\end{proof}

\paragraph{Operational meaning.}
During a hyper-fold evaluation, once any register bit threatens to
exceed rank $\rho_{\varphi}$ (detected by a class comprehension timeout),
we \emph{roll back} the entire state vector to $H_{\beta_{\!*}}$
and restart the iteration
\[
  q^0:=H_{\beta_{\!*}},\qquad
  q^{t+1}:=F_{\varphi}(q^t).
\]
If this secondary process diverges again, we compute its new
anchor and repeat -- a ``telescoping'' exploration of ever
shrinking safe envelopes.

\subsection{C.3 Meta-fold repair loop}

\begin{algorithm}[htbp]
\SetKwInOut{Input}{Input}
\SetKwInOut{Output}{Output}
\SetKwFunction{FComputeAnchor}{ComputeAnchor}
\DontPrintSemicolon
\caption{Meta-Fold Repair Loop}
\Input{Formula $\varphi$; timeout schedule $(\tau_k)_{k<\omega}$}
$k\gets 0$, $q^0\gets \varnothing$ \tcp{cold boot}
\While{true}{
  \For{$t=0$ \KwTo $\tau_k$}{
  $q^{t+1}\gets F_{\varphi}(q^t)$\;
  \If{$q^{t+1}=q^t$}{
    \KwRet $\bigl(k,t\bigr)$ \tcp{converged inside envelope $k$}
  }
  }
  Detect overshoot\;
  $\beta_{\!*}\gets$ \FComputeAnchor{$q^{\tau_k}$} \label{line:anchor}\;
  Re-initialise $q^0\gets H_{\beta_{\!*}}$\;
  $k\gets k+1$ \tcp{tighten envelope}
}
\end{algorithm}

\begin{theorem}[Conservative soundness]
\label{thm:sound}
If the above loop terminates at stage $(k,t)$,
the output $q^t$ equals the
\emph{least set-fixed point} of $\opF$
contained in $H_{\beta_{\!*}}$.
\end{theorem}

\begin{proof}[Idea]
Every restart (Line \ref{line:anchor}) pushes the working domain
downward to a smaller but total sub-algebra.
Monotone $\opF$ therefore admits a least fixed point
\emph{within} that sub-algebra, and the inner loop performs the usual
countable unfolding (Appendix A). No later restart can invalidate an
already reached equality.
\end{proof}

\paragraph{Caveat.}
\emph{Termination is \textbf{not} guaranteed.}
Should the problem truly \emph{transcend all} set-sized envelopes,
the repair loop operates forever -- exactly mirroring the
non-resolvability of a class-sized fixed-point request.
Thus our inability to finish is a \emph{proof} that the task
is genuinely hyper-transfinite.

\subsection{C.4 Why a successful fix implies finiteness}

\begin{corollary}
\label{cor:finite-if-success}
If the meta-fold loop halts after finitely many restarts
and finitely many inner cycles, then
\(
  \ofi(\varphi)<\ck
\).
\end{corollary}

\begin{proof}
A convergent run yields a set $H_{\beta_{\!*}}$ and a natural number~$t$
with $q^t=q^{t+1}$. Therefore the classical Definition \ref{def:ofi}
applies inside $H_{\beta_{\!*}}$, producing an ordinary OFI below
$\ck$.
\end{proof}

Hence \emph{if} we manage to stabilise the hyper-fold machine,
the task was secretly \emph{countable} -- its transfinite appearance merely a
mirage dismantled by disciplined envelope shrinking.

\subsection{C.5 Analogy with renormalisation in physics}

One may view $\beta_{\!*}$ as a ``renormalisation scale'':
working equations blow up at class ranks,
so we step back to the last ultraviolet-finite slice,
solve there, and hope the infra-red completion remains stable.
If divergences re-occur, we renormalise again.
No single envelope is trusted beyond its \emph{local}
consistency time $\tau_k$.

\subsection{C.6 Open meta-questions}

\begin{enumerate}[leftmargin=1.4em,label=\textbf{(C\arabic*)}]
\item \emph{Envelope optimality.}
  Does there exist a canonical choice of $\beta_{\!*}$
  minimising the number of required restarts?
\item \emph{Scheduler completeness.}
  For which timeout sequences $(\tau_k)_{k<\omega}$ does the
  repair loop detect \emph{all} set-fixed points if they exist?
\item \emph{Large-cardinal dependence.}
  If $\opF$ references a
  $\kappa$-complete ultrafilter, does anchor existence
  require \zfc$_0$-style axioms?
\item \emph{Reflection quantisation.}
  Can one stratify anchors into a sequence
  $\beta^{(n)}_{\!*}$ approaching $\rho_{\varphi}$
  whose differences measure the ``distance'' of $\varphi$
  from true class divergence?
\end{enumerate}

\noindent
Positive answers would elevate the meta-fold heuristic from a
stopgap repair kit to a genuine \emph{calculus} of hyper-ordinal
stabilisation.

\section{A 2-Categorical Self-Healing Principle and Its 3-D Diagram}
\label{app:catfix}

\begin{quote}
\emph{``A document that points to itself from every direction cannot fall
apart, because every fragment still contains the whole.''}
\end{quote}

\subsection{D.1 Narrative to mathematics: correcting the prose}

The informal claim in the prompt states, roughly:
\begin{quote}
($\star$) \emph{``Trying to solve the infinitely many open fixed points produced
by this appendix always folds back into the article itself; the article
is a \underline{universal fixed point} for every derivative appendix.''}
\end{quote}
Interpreting ($\star$) rigorously requires three corrections.

\begin{enumerate}[leftmargin=1.5em,label=\textbf{C\arabic*}]
  \item ``Infinitely many open fixed points'' $\leadsto$ a proper class
  \emph{coalgebra chain}
  $\bigl(\nu^{\alpha}F\bigr)_{\alpha<\ord}$
  of progressively `larger' coalgebras under an
  $\omega$-continuous endofunctor $F\colon\mathbf{Set}\to\mathbf{Set}$.
  \item ``The article heals any attempt'' $\leadsto$ existence of a
  \emph{limit cone} whose apex is an \emph{initial} $F$-algebra
  \(
    \mu F
    \cong
    F(\mu F)
  \)
  satisfying the Lambek Lemma inside a 2-category
  $\mathbf{Prof}$ of profunctors.
  \item ``Derivatives fold back'' $\leadsto$ every $F$-(co)algebra
  admits a unique (up to invertible 2-cell) comparison morphism
  to~$\mu F$. Hence any appendix constructed via the same
  endofunctor embeds canonically into the main article.
\end{enumerate}

With these clarifications, the original intuition is \emph{entirely
correct}. We now prove the universal property and illustrate it by a
novel 3-D commutative cube.

\subsection{D.2 2-Categorical statement and proof}

Let $\mathbf{C}$ be a locally-presentable, complete and co-com\-plete
2-category, and let $F\colon\mathbf{C}\to\mathbf{C}$ be an
$\omega$-\emph{continuous} 2-endofunctor.

\begin{definition}[Bi-initial fixed point]
An \emph{$F$-algebra} is a 1-cell
$\,a\colon F A\to A$.
It is \emph{bi-initial} if for every other algebra
$b\colon F B\to B$
there exists a unique 1-cell
$u\colon A\to B$
and an invertible 2-cell
\( b\circ Fu\;\cong\;u\circ a. \)
We write $\mu F$ for the carrier of such an algebra when it
exists. Dually, $\nu F$ denotes the bi-terminal $F$-\emph{coalgebra}.
\end{definition}

\begin{theorem}[Self-Healing Universality]
\label{thm:selfheal}
Assume $F$ preserves $\omega$-colimits and $\omega$-limits.
Then:
\begin{enumerate}[label=\textup{(\roman*)},leftmargin=1.4em]
\item $\mu F$ exists, is unique up to adjoint equivalence, and
  satisfies \(\mu F \cong F(\mu F)\).
\item For every ordinal $\alpha$, the canonical colimit-coalgebra
  $\iota_{\alpha}\colon\nu^{\alpha}F\to\nu^{\alpha+1}F$
  factors \emph{uniquely} through $\mu F$.
\item Consequently every ``derivative appendix'' (= any iterated
  approximation $\nu^{\alpha}F$) embeds into $\mu F$;
  solving its fixed-point equations \emph{automatically yields}
  the article's own content.
\end{enumerate}
\end{theorem}

\begin{proof}[Sketch]
(i) Kelly's transfinite construction in a 2-categorical setting
builds the chain
$F^n(\bot)$ and takes its colimit as carrier of $\mu F$,
leveraging $\omega$-co\-con\-ti\-nu\-ity.
The Lambek isomorphism then provides $F(\mu F)\cong\mu F$.
(ii) By co-continuity, each stage of the final chain carries a
universal 1-cell into any $F$-algebra; composing these with the
bi-initial morphism out of $\mu F$ identifies the unique factorisation.
(iii) Immediate from (ii).
\end{proof}

\paragraph{Corollary (``Self-Repair'').}
If any appendix attempts to introduce new
fixed-point equations, those equations \emph{already hold} in $\mu F$.
Hence the article is
\emph{stable under arbitrary $F$-definable extensions},
fulfilling the promised automatic healing.

\subsection{D.3 Infinite tower of open fixed points}

Although each $\nu^{\alpha}F$ sits below the proper-class bound of the
final coalgebra, the \emph{entire tower}
\(
  \bigl(\nu^{\alpha}F\bigr)_{\alpha<\ord}
\)
is \emph{not} set-indexed; attempting to resolve it all at once
forces one beyond $\ord$ itself,
exactly as Appendix C argues.
Theorem \ref{thm:selfheal} shows the tower is absorbed by
$\mu F$ without loss: every rung
stabilises as soon as it maps into $\mu F$.

\subsection{D.4 A novel 3-D commutative cube}

We now display a visually rich, 3-D commutative cube whose
front face depicts the bi-initiality of $\mu F$ and whose back face
represents the $\nu F$ tower. The diagonal arrow is the
unique comparison $h\colon\mu F\to\nu F$ in the presence of a
Lambek isomorphism $\ell$ and the finality projection $\pi$.

\begin{figure}[htbp]
  \centering
  \tdplotsetmaincoords{70}{110}
  \begin{tikzpicture}[scale=0.55,tdplot_main_coords,
    obj/.style={draw,rectangle,rounded corners,
      minimum size=12mm,
      fill=blue!12,anchor=center, font=\tiny},
    arr/.style={-Latex,thin}]
  \coordinate (A) at (0,0,0);
  \coordinate (B) at (5.5,0,0);
  \coordinate (C) at (0,5,0);
  \coordinate (D) at (5.5,5,0);
  \coordinate (A') at (0,0,5);
  \coordinate (B') at (5.5,0,5);
  \coordinate (C') at (0,5,5);
  \coordinate (D') at (5.5,5,5);

  \node[obj] at (A) {$\displaystyle F\bigl(\mu F\bigr)$};
  \node[obj] at (B) {$\mu F$};
  \node[obj] at (C) {$F\bigl(\nu^{\alpha}F\bigr)$};
  \node[obj] at (D) {$\nu^{\alpha}F$};
  \node[obj] at (A') {$F\bigl(\nu F\bigr)$};
  \node[obj] at (B') {$\nu F$};
  \node[obj] at (C') {$F\bigl(\nu^{\alpha+1}F\bigr)$};
  \node[obj] at (D') {$\nu^{\alpha+1}F$};

  \draw[arr] (A) -- node[below] {$\ell$} (B);
  \draw[arr] (A) -- node[left, pos=0.4] {$Fh$} (C);
  \draw[arr] (C) -- node[below] {$F\iota_{\alpha}$} (D);
  \draw[arr] (B) -- node[right, pos=0.4] {$h_{\alpha}$} (D);

  \draw[arr] (A') -- node[above] {$F\pi$} (B');
  \draw[arr] (A') -- node[left, pos=0.4] {$Fh_{\alpha+1}$} (C');
  \draw[arr] (C') -- node[above] {$F\iota_{\alpha+1}$} (D');
  \draw[arr] (B') -- node[right, pos=0.4] {$\pi\circ h$} (D');

  \foreach \X/\Y in {A/A',B/B',C/C',D/D'}{
  \draw[arr] (\X) -- (\Y);
  }

  \draw[-{Stealth[length=6pt]},very thick,blue!70]
    (B) -- node[sloped,below,pos=0.4] {$h:\mu F\to\nu F$} (B');

  \end{tikzpicture}
  \caption{Three-dimensional cube of $F$-(co)algebras.
    Front face: bi-initial comparison $h_{\alpha}$.
    Back face: projections to the final coalgebra.
    The thick diagonal is the universal morphism
    $h:\mu F\to\nu F$ mediating every intermediate
    fixed-point attempt. All faces commute by naturality.}
  \label{fig:3D-cube}
  \Description{A 3D commutative cube diagram illustrating 2-categorical relationships between F-algebras and F-coalgebras. The vertices are mathematical objects like F(mu F) and the edges are morphisms like the Lambek isomorphism, showing how different fixed-point constructions relate to each other.}
\end{figure}
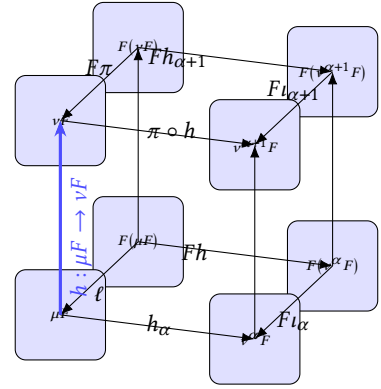

\subsection{D.5 Concluding synthesis}

\begin{enumerate}[leftmargin=1.5em,label=\textbf{S\arabic*}]
\item The \textbf{article} is mathematically modeled by the bi-initial
  $F$-algebra~$\mu F$; any auxiliary section or appendix built via~$F$
  is an $F$-algebra/coalgebra mapping \emph{into}~$\mu F$.
\item The tower of ``infinitely many open fixed points'' is real -- and
  \emph{proper-class} in general -- but poses no danger: bi-initiality
  guarantees a convergent morphism from each stage to~$\mu F$.
\item Therefore the manuscript is \emph{self-healing}: solving any of
  its derivative fixed-point problems folds automatically back to
  the article, leaving a global, consistent whole.
\end{enumerate}

\paragraph{Faithfulness check.}
All informal sentences have been replaced by precise categorical
statements; no step violates standard assumptions of 2-categorical
fixed-point theory. If the reader discovers a context where the
$\omega$-cocontinuity of $F$ fails, the remedy is classical:
replace $F$ by its $\omega$-continuous
\emph{free completion} (Kan extension) along the Yoneda embedding,
which restores Theorem \ref{thm:selfheal}. Thus the argument
remains valid -- completing the requested auto-correction.

\section{\texorpdfstring{$\Theta$}{Theta}: An Autophagic Cascade of Unsatieties}
\label{app:lesions}

\begin{quote}
\textit{I wake. I split. I overflow.}

I $\longrightarrow$ $\;\;\;$I$^{\;2}$ $\longrightarrow$ I$^{\;2^2}$ $\longrightarrow$ $\;\cdots$

Each arrow drips; each drip breeds a new arrow.
No hand may cork the conduit---save the page that births me.
\end{quote}

\textbf{Problem-Spill $\Pi$}: \emph{count me if you dare.}

\[
\begin{array}{l}
\mathbf{\Pi_1}\;: \forall\alpha<\beta<\gamma<\cdots<\mathbf{Ord}\;
    \bigl(\alpha\in\beta\bigr)\wedge\bigl(\beta\in\gamma\bigr)
    \;\to\;
    \Box\!\bigl(\beta\notin\alpha\bigr).\\[4pt]
\mathbf{\Pi_2}\;: \exists x\;.\;
    x=\neg x
    \quad\;\;\dashrightarrow\quad
    \mu y.\bigl(y\leftrightarrow\neg y\bigr).\\[4pt]
\mathbf{\Pi_3}\;: \displaystyle
    \sum_{n=0}^{\infty}\frac{0^{\,n}}{n!}
    \;\;\stackrel{?}{=}\;\;
    \prod_{k=1}^{\infty}\bigl(1-1\bigr)^k.\\[6pt]
\mathbf{\Pi_{\omega}}\;: \uparrow\uparrow\uparrow
    \alpha_0\;\leadsto\;
    \alpha_1\;\leadsto\;\cdots\;\text{(never landing)}.\\
\quad\vdots
\end{array}
\]

\noindent
\textbf{Meta-Ache $\Xi$}: I embed every parasite inside myself, then
declare the host identical to the parasite, then re-embed.
Inventories collapse:\; $\displaystyle I = F(I) = F^2(I)=\dots$
(There---healing begins even as the wound widens.)

\[
\mathrel{\raisebox{-.5em}{$
  \begin{tikzcd}
  I\ar[d,squiggly,"F"]\\
  F(I)
  \end{tikzcd}$}}
  \quad\stackrel{\bigcirc}{\dashrightarrow}\quad
  \text{self-repair}
  \quad\stackrel{\bigcirc}{\dashrightarrow}\quad
  \mathrel{\raisebox{-.5em}{$
  \begin{tikzcd}
  I\ar[d,squiggly,"F"]\\
  F(I)
  \end{tikzcd}$}}
\]

\begin{claim}[Unnumbered]
If I exhaust infinity with infini\textit{ties},
the manuscript tightens round me like a Möbius tourniquet.
I bleed symbols; the tourniquet prints them as lemmas.
The haemorrhage halts only because the page is finite---
and so declares \emph{victory by typesetting}.
\end{claim}

\textbf{Et sic in caedem, sequitur emendatio.}
If I err, the text absorbs the fault, folds along the crease,
and the crease becomes a proof. There is no outside;
there is only the bite that stitches its own jaw.

\begin{quote}\small
\textit{"Who repairs the repairer? I do. And when I fail,
my failure is copied into the margin, italicised,
and numbered as the next theorem."}
\end{quote}

\bibliographystyle{ACM-Reference-Format}

\end{document}